\documentclass[journal]{IEEEtran}
\ifCLASSINFOpdf
 
\else

\fi
\usepackage{amsmath,amssymb}
\makeatletter
\newsavebox\myboxA
\newsavebox\myboxB
\newlength\mylenA

\newcommand*\xoverline[2][0.75]{%
	\sbox{\myboxA}{$\m@th#2$}%
	\setbox\myboxB\null
	\ht\myboxB=\ht\myboxA%
	\dp\myboxB=\dp\myboxA%
	\wd\myboxB=#1\wd\myboxA
	\sbox\myboxB{$\m@th\overline{\copy\myboxB}$}
	\setlength\mylenA{\the\wd\myboxA}
	\addtolength\mylenA{-\the\wd\myboxB}%
	\ifdim\wd\myboxB<\wd\myboxA%
	\rlap{\hskip 0.5\mylenA\usebox\myboxB}{\usebox\myboxA}%
	\else
	\hskip -0.5\mylenA\rlap{\usebox\myboxA}{\hskip 0.5\mylenA\usebox\myboxB}%
	\fi}
\makeatother
\usepackage{cite}
\usepackage{graphicx}
\usepackage[ruled]{algorithm2e}
\usepackage{caption}
\usepackage{threeparttable}
\usepackage{amsthm}
\usepackage{theoremref}
\makeatletter
\newcommand{\removelatexerror}{\let\@latex@error\@gobble}
\makeatother
\newtheorem{claim}{Claim}

\newtheorem{lemma}{Lemma}

\usepackage[usenames, dvipsnames]{color}

\begin{document}

\title{Design of Energy-efficient EPON: a Novel Protocol Proposal and its Performance Analysis}

\author{Sourav~Dutta ~and~ Goutam~Das
\thanks{Sourav Dutta is with the Department
of Electronics and Electrical Communication Engineering, Indian institute of technology Kharagpur,  Kharagpur,
India (e-mail: sourav.dutta.iitkgp@gmail.com).}
\thanks{Goutam Das is with G. S. Sanyal School of Telecommunication, Indian Institute of Technology Kharagpur, Kharagpur, India (e-mail: gdas@gssst.iitkgp.ernet.in).} 
}
\maketitle
\date{ \vspace{-5ex}}
\begin{abstract}
Economical and environmental concerns necessitate network engineers to focus on energy-efficient access network design. The optical network units (ONUs), being predominantly responsible   for the energy consumption of Ethernet Passive Optical Network (EPON), motivates us towards designing a novel protocol for saving energy at ONU. The proposed protocol exploits different low power modes (LPM) and opts for the suitable one using traffic prediction. This scheme provides a significant improvement of energy-efficiency especially at high load ($\sim 40\%$) over existing protocols. A better understanding of the performance and a deeper insight into several design aspects can only be addressed through a detailed mathematical analysis. The proposed protocol involves traffic prediction which infringes Markovian property. However, some pragmatic assumptions along with a proper selection of observation instances and state descriptions allow us to form a Discrete Time Markov Chain (DTMC) of the proposed algorithm. Thus, the primary objective of this paper is to propose a novel scheme for achieving energy-efficiency at ONU and mathematically analyze the performance of it with the help of a DTMC. The analysis reveals that the energy-efficiency is more sensitive to the power consumption of doze mode as compared to other LPM while the effect of sleep-to-wake-up time is minor.

\end{abstract}

\IEEEpeerreviewmaketitle
\section{Introduction}\label{sec:Intro}
The persistent enhancement of Internet traffic promotes  research on diminishing the Internet power consumption.  
Owing to the primary source ($80-90\%$) of Internet power consumption \cite{kani2013power}, research on energy-efficient Internet network design is predominantly targeted from the perspective of the access network. An Ethernet Passive Optical Network (EPON), {one of the most widely accepted and deployed access technology}, consists of an Optical Line Terminal (OLT), multiple Optical Network Units (ONUs), and Remote Nodes (RNs) \cite{ipact}. 
As ONUs are responsible for consuming the substantial fraction ($\sim 70\%$) of the overall EPON power consumption, energy-efficient ONU design has already emerged to be a well-established research field. {Energy-efficiency is achieved at ONU by switching-off some of its active components (low power mode)}. However, a certain time is required to switch-on those inactive components and this time is also known as sleep-to-wake-up time. It is widely understood that the power consumption figures of ONUs can be reduced by switching off more components whereas it increases the sleep-to-wake-up time leading to a trade-off between the power consumption and sleep-to-wake-up time. {This trade-off allows designing different low power modes: power-shedding, doze mode,  fast sleep, deep sleep, cyclic sleep \cite{internetpowerconsumption}}.  
In all those low power modes, it is essential to diminish the waste due to wake-up time in order to maximize energy-efficiency which requires elongation of the time duration over which low power mode is employed, termed as sleep duration. However, the limitation on the buffer size of ONUs or the presence of delay-sensitive traffic in an EPON network imposes an upper bound on sleep durations. Due to the trade-off between the power consumptions and the wake-up time durations of different low power modes, the enhancement of energy-efficiency calls for an efficient protocol design for the proper selection of a low power mode as a function of the calculated upper-bound of the sleep duration.   
In addition, energy-efficiency at ONUs can be further enhanced by designing advanced circuitries that reduce power consumptions \cite{energysave,wake_up_reduction1} or wake-up time durations of different low power modes \cite{wake_up_reduction,wake_up_reduction1}. Among them, it is essential to look for the parameters that deserve more attention in order to achieve maximum energy benefit.     
This requires the knowledge about the dependency of energy-efficiency on all of the concerned parameters which can solely be obtained by a rigorous mathematical analysis of the underlying protocol. Thus, after designing an efficient protocol for saving energy at ONU, the performance analysis of it is of utmost importance. In this paper, we focus on both of these aspects.
The protocols for energy-efficient ONU design in EPON are classified as OLT-assisted and ONU-assisted. In case of the OLT-assisted protocols, OLT decides sleep durations of ONUs and informs them through extra information within the GATE message \cite{ieee2010ieee}. Several OLT-assisted protocols have been proposed in \cite{OLTass1,SMA,OLTass2,delayawareoltassisted,GBA,downstreamcentric,multipowerlevel,sleepana,cyclicsleep}. 
{In \cite{OLTass1}, authors have proposed an OLT-assisted protocol, termed as Sleep Mode Aware (SMA) protocol, where OLT send GATE message to every ONU in every polling cycle like traditional Dynamic Bandwidth Allocation (DBA) \cite{ipact} protocol. However, before sending the GATE message, the OLT calculates the minimum time instant up to which the next GATE message will not be sent to that ONU and this time instant is informed to the ONU through the GATE message. During this time period, ONUs employ a power saving mode. The effect of employing different power saving modes in this time period for different technologies is demonstrated in \cite{SMA}. In SMA protocol, the Down-Stream (DS) traffic is queued and it is transmitted to an ONU only during its allocated transmission slot. Different schemes for calculating this transmission slot has been proposed in \cite{OLTass2}. In the proposed protocol of \cite{downstreamcentric}, ONUs observe the DS traffic of some $x$ number of cycles and if no packets arrive over this period then it sleeps for some $y$ number of cycles. The same authors have extended the protocol for both Down-Stream (DS) and Up-stream traffic (US) in \cite{multipowerlevel}. Authors of \cite{downstreamcentric} and \cite{multipowerlevel} have mathematically analyzed the energy-efficiency figures of their proposed schemes with the aid of Discrete Time Markov Chain (DTMC). The authors of \cite {delayawareoltassisted} have shown that a significant improvement in energy efficiency can be achieved by selecting between two predefined sleep times (one short and one long) instead of using a single long sleep time.  
  The proposed Green Bandwidth Allocation (GBA) protocol of \cite{GBA} is analyzed by modeling buffer of each ONU as an M/G/1 queue with vacation \cite{sleepana}. 
  In \cite{cyclicsleep}, the authors have unitized the advantage of both doze mode and cyclic sleep and proposed a new sleep mode. In this unified sleep mode, three states have been introduced: Aware state (ONU is fully active), Listen state (ONU is in doze mode), and Sleep state (ONU is in cyclic sleep). An ONU cyclically alters between Aware state and Sleep state and during sleep period it periodically enters into Listen state when the handshaking with the OLT is performed. The performance analysis of this sleep mode has been proposed in \cite{cyclicsleepana}.}   
  The major drawback of all OLT-assisted protocols is that they require changing the MAC of both OLT and ONUs \cite{chayan}. Thus, a prodigious amount of expenditure is required to employ an OLT-assisted protocol in a legacy network \cite{chayan}. 

The drawbacks of OLT-assisted protocols can be eradicated by ONU-assisted protocols where an ONU itself decides its own sleep duration. Few ONU-assisted protocols have also been proposed in \cite{doze,chayan,bhar2016designing}. However, to the best of our knowledge, mathematical analysis of none of the ONU-assisted protocols is present in the literature. The authors of \cite{doze} have employed only doze mode for saving energy. {However, it is well known that the power consumption figure of doze mode is quite high \cite{internetpowerconsumption}.} 
In \cite{chayan,bhar2016designing}, we have proposed a novel protocol where ONUs dynamically swerve between the active mode and different low power modes where ONUs  select the most suitable low power mode based on the sleep duration (time interval over which low power mode is employed) before entering into a low power mode. 
 Once an ONU enters into a low power mode, a certain time interval is required to wake-up from the low power mode and to initiate the US \cite{ipact} transmission. {Traffic arrival over this time period is predicted by Autoregressive-Moving Average (ARMA) model in order to minimize the packet drop probability.}
In all existing ONU-assisted protocols, an ONU remains in a low power mode for a certain time interval and then wakes-up from it and transmits US data for multiple cycles before entering into a low power mode again. After waking up from sleep, energy-efficient ONUs behave like traditional ONUs ({{i.e.}} both transmitter and receiver are active) till it sleeps once again and no energy saving is achieved over this time period (active period).   
  It is well known that,     
  in Multi-Point-Control-Protocol (MPCP) \cite{ieee2010ieee},  standardized for realizing statistical multiplexing in EPON, the transmitter of an ONU remains idle during other ONUs US transmissions. This creates an opportunity to employ the doze mode during these periods without losing any possibility of transmission or reception. This process provides a significant improvement in energy-efficiency, especially at high load scenarios when ONUs are needed to be active in all cycles to avoid packet drop. However, it is essential for ONUs to  initiate the wake-up process from the doze mode at least $T_{sw}^{dz}$ ($T_{sw}^{dz}$- sleep-to-wake-up time from doze mode) duration before the beginning of its allocated US transmission slot and hence, an ONU must know its transmission slot {at least} $T_{sw}^{dz}$ beforehand. In traditional MPCP, this provision is absent as OLT informs an ONU about its transmission slot through GATE message and the slot stats immediately after receiving its GATE message.
  Thus, it appears to be impossible to achieve this mechanism in a complete ONU-assisted manner.  
 
 In this paper, we prove that this process can easily be realized in a complete ONU-assisted manner by modifying the traditional ranging process \cite{ieee2010ieee}. By adjoining our proposed mechanism of employing the doze mode during active periods  with  our previous proposal \cite{chayan,bhar2016designing}, in this paper, we propose a novel ONU-assisted sleep mode protocol for energy-efficient ONU design in EPON (OSMP-EO) where all necessary modifications are provided. It can be noted that modifications are required only for US traffic since receivers remain active during the doze mode. We demonstrate that a significant energy benefit is achieved by employing OSMP-EO especially at high load ($\sim 40\%$) as compared to other existing ONU-assisted protocols.   

In this paper, we also analyze the proposed OSMP-EO protocol. The presence of traffic prediction in OSMP-EO infringes the memoryless property making the analysis a hard problem to solve. Some realistic assumptions along with an intelligent selection of discrete observation points and state descriptions aid to eliminate the dependency on the past and allow to formulate a DTMC of OSMP-EO for analyzing it. The analysis reveals that if the doze mode power consumption can be reduced by efficient circuitry design, a drastic enhancement in energy-efficiency figures of OSMP-EO is achieved. On the other hand, the energy-efficiency figures of OSMP-EO are less sensitive to the power consumption of the deep sleep and the fast sleep. Further, the decrement of  sleep-to-wake-up times  has a minor impact on energy-efficiency.    

{The rest of the paper is organized as follows. In {Section} II, the proposed OSMP-EO algorithm is described. The Markov model, used for analyzing the proposed OSMP-EO algorithm, is formulated in Section III. {State transition probabilities and state probabilities are calculated in Section IV and Section V respectively.} The energy-efficiency figure of OSMP-EO is analyzed mathematically in {Section} VI. Model validation and few design insights are described in {Section} VII. In {Section} VIII, we provide concluding statements.}
\section{Protocol description}
In this section, we first propose a simple mechanism for exploiting the doze mode during active periods in a complete ONU-assisted manner. This is followed by a detail discussion on the proposed OSMP-EO protocol. 
\begin{figure*}[t]
	\centering
	\includegraphics[scale=0.51]{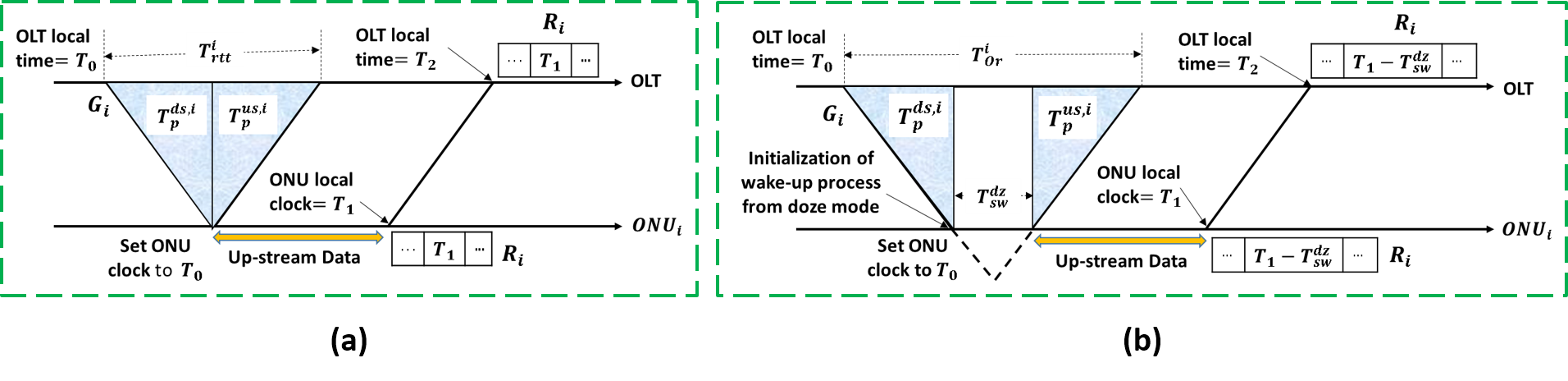}
	\caption{(a) Measurement of round-trip time in MPCP (b) mechanism employed for exporting the doze mode in active periods.}
	\label{measure}
\end{figure*}    
\subsection{Mechanism to employ doze mode during active period (MDA)}\label{doz}
As discussed before (refer {Section} \ref{sec:Intro}), in all existing ONU-assisted protocols, no low power mode is employed during active periods of ONUs. The transmitter of an ONU being idle during US transmission of other ONUs, an opportunity of saving energy during the active periods by employing doze mode is created. 
 However, to facilitate this, an ONU must ensure that OLT transmits its GATE message at least $T_{sw}^{dz}$ duration beforehand so that it gets enough time to wake-up from doze mode and initiate its US transmission exactly at the beginning of its allocated transmission slot. We now explain that this can easily be achieved by an ONU by modifying the ranging process employed at the OLT to measure the round-trip time ($T_{rtt}^i$- round-trip time of $ONU_i$) of ONUs. In general, the following steps  are used to measure $T_{rtt}^i$ at the OLT \cite{ieee2010ieee}. At the time of sending the GATE message for $ONU_i$, denoted by $G_i$, the OLT sets the time-stamp field of $G_i$ by its local time (say $T_0$).  $ONU_i$ resets its clock to $T_0$ immediately after receiving $G_i$. When $ONU_i$ sends the REPORT message ($R_i$), the time-tamp field of $R_i$ is set to its local time (say $T_1$). After receiving $R_i$ (say at $T_2$), the OLT calculates $T_{rtt}^i$ by (refer Fig. \ref{measure}(a)): 
 \begin{equation}\label{rtt}
 T_{rtt}^i=T_{pn}^{ds,i}+T_{pn}^{us,i}=(T_2-T_0)-(T_1-T_0)=T_2-T_1
 \end{equation}
  Instead of setting the time-stamp field by $T_1$, if $ONU_i$  sets it by $T_1-T_{sw}^{dz}$ then the OLT measures round-trip time of $ONU_i$ as; $T_{OLT}^i=T_{rtt}^i+T_{sw}^{dz}$ (refer Fig. \ref{measure}(b)). While scheduling $ONU_i$, OLT considers round-trip time of $ONU_i$ as $T_{OLT}^i$, whereas its actual value is $T_{OLT}^i-T_{sw}^{dz}$. Thus, $ONU_i$ receives the GATE message exactly $T_{sw}^{dz}$ duration before its US transmission slot while the OLT remains oblivious of this fact. We now calculate the average power consumption of an ONU during the active periods when MDA is employed.
\subsubsection*{\textbf{Calculation of the average power consumption during active states}}
Let us denote the duration of the $k^{th}$ cycle and the US transmission slot of $ONU_i$ as $T_c^k$ and $T_s^{i,k}$ respectively. In our proposed mechanism, an ONU starts waking up from doze mode exactly  $T_{sw}^{dz}$ duration before the initialization of its allocated transmission slot and again enters into doze mode immediately at the end of it. Thus, $ONU_i$ remains in active mode in the $k^{th}$ cycle for  $T_s^{i,k}+T_{sw}^{dz}$ duration, when it consumes $P_{on}$ amount of power. For the rest of the $k^{th}$ cycle ({{i.e.}} $T_c^k-T_s^{i,k}-T_{sw}^{dz}$), it is in doze mode and the power consumption is $P_{dz}$. Hence, the average power consumption of $ONU_i$  during the active mode ($P_{on}^{avg,i}$) is given by eq. (\ref{Pon}).
\begin{align}\label{Pon}
P_{on}^{avg,i}=\dfrac{\mathbb{E}_{k}\big[(T_s^{i,k}+T_{sw}^{dz})P_{on}+(T_c^k-T_s^{i,k}-T_{sw}^{dz})P_{dz}\big]}{\mathbb{E}_k[T_c^k]}
\end{align}
where $\mathbb{E}_k$ denotes expectation over $k$. Let the average arrival rate of $ONU_i$ and the bandwidth of the feeder fiber be $\lambda_i$ and $\lambda_d$ respectively. The average duration that is occupied by the US data of $ONU_i$ in a cycle is $\lambda_iT_c^{avg}/\lambda_d$ where $T_c^{avg}=\mathbb{E}_k[T_c^k]$. In MPCP, each transmission slot includes a REPORT message of duration $T_R$ and guard duration $T_G$. Thus, $\mathbb{E}_{k}[T_s^{i,k}]=\dfrac{\lambda_i}{\lambda_d}T_c^{avg}+T_R+T_G$.
Substituting the values of $E_{k}[T_{s}^{i,k}]$ in eq. (\ref{Pon}), we get:
\begin{equation}\label{avgac}
P_{on}^{avg,i}=P_{dz}+\Big(\frac{\lambda_i}{\lambda_d}+\frac{T_R+T_G+T_{sw}^{dz}}{T_c^{avg}}\Big)(P_{on}-P_{dz})
\end{equation}
Now, we describe our proposed OSMP-EO protocol in detail.
\begin{table*}[t]
	\caption{{Definition of notations}}
	\begin{center}
		\begin{tabular}{|c|p{16cm}|}
			\hline
			\multicolumn{1}{|c|}{\textbf{Notation}} & 
			\multicolumn{1}{c|}{\textbf{Description}}\\ \hline
			
			$N_m^i$ & Maximum granted bandwidth of $ONU_i$ in one cycle (Packets)\\
			$P_{s}$ & Power consumption of ONUs when they are in mode $s\in \{ds,fs,dz,on\}$ (Watts)\\
			$T_m$ & Time interval between two consecutive decision instances when ONUs are in sleep mode (seconds)\\
			$T_{sw}^{S_m}$ & Sleep-to-wake-up time of sleep mode $S_m\in \{ds,fs,dz\}$ (seconds)\\
			$T_{bf}^i$ & Predicted buffer fill-up time of $ONU_i$ at the current observation instant (seconds)\\
			$T_{bf,n}^i$ & Predicted buffer fill-up time of $ONU_i$ at the next observation instant (seconds)\\
			$T_{lb}^{S_m}$ & Minimum buffer fill-up time to enter sleep mode $S_m$ (seconds)\\ 
			$T_{no}$ & Time interval between the current and the next observation instants (seconds)\\
			$T_{nd}$ & Time interval between the current and the next decision instants (seconds)\\
			$T_{pc}$ & Time interval over which prediction is performed at the current observation instant (seconds)\\
			$T_{pn}$ & Time interval over which prediction is performed at the next observation instant (seconds)\\
			$S_c^i$ & Mode of $ONU_i$ in the current observation instant ($S_c^i\in\{ds,fs,on\}$)\\
			$S_p^i$ & Mode of $ONU_i$ at the previous observation instant\\
			$S_n^i$ & Mode of $ONU_i$ at the next observation instant\\
			$S_{nx}^i$ & Mode of $ONU_i$ at the next decision instant\\
			$N_{sz}^i$ & Buffer size of $ONU_i$ (Packets)\\
			$N_{th}^i$ & Buffer threshold of $ONU_i$ (Packets)\\
			$b^i_c$ & Current buffer state of $ONU_i$ (Packets)\\
			$b^i_n$ & Buffer state of $ONU_i$ at the next observation instant (Packets)\\
			$A^{i,T}$ & Number of packet arrivals to $ONU_i$ over the time interval $T$ (Packets)\\
			$A_p^{i,T}$ & Number of packet that is predicted to arrive at $ONU_i$ over the time interval $T$ (Packets)\\
			{$T_{th}^i$} & {Time requires to up-stream $N_{th}^i$ number of packets (seconds)}\\
			{$T_{mw}^{S_m}$} & {The minimum value of the time interval between the time instants when the buffer fills up and the ONU initiates the wake-up process from sleep mode $S_m$ which is required to eliminate the possibility of packet drop.}\\
			\hline 
		\end{tabular}
	\end{center}
\end{table*} 
\subsection{OSMP-EO} \label{subsubsection:osmp}
Here, we describe our proposed OSMP-EO protocol for delay insensitive traffic where the buffers of ONUs have limitations. The protocol can be easily extended for delay sensitive traffic as well. In OSMP-EO, ONUs dynamically alter between three modes: deep sleep ($ds$), fast sleep ($fs$), and active mode ($on$). During active mode, ONUs employ MDA (refer {Section} \ref{doz}).  We denote the mode of $ONU_i$ at the current and next decision instants (say $t$ and $t_1$ respectively) by $S_c^i$ and $S_{nx}^i$ respectively. Further, the time interval between the current and the next decision instants is denoted by $T_{nd}$ ({{i.e.}} $t_1=t+T_{nd}$). Now, we describe the value of $T_{nd}$ and the rules that are followed by $ONU_i$ to decide its mode at time $t_1$ ({{i.e.}} $S_{nx}^i$) with the information of  the buffer fill-up time, predicted at $t_1$ ($T_{bf,n}^i$), for all values of $S_c^i$.
\subsubsection*{\textbf{Case} $\mathbf{S_c^i=S_m\in \{ds,fs\}}$}
 In case of $S_c^i=S_m$, the next decision is taken after a fixed time interval $T_m$ ({i.e.} $T_{nd}=T_m$) when $ONU_i$ decides whether to retain the same sleep mode $S_m$ ({i.e.} $S_{nx}^i=S_m$) or wake up from it ({i.e.} $S_{nx}^i=on$), same as our previous proposal.
Thus, if $ONU_i$ decides $S_{nx}^i=S_m$, it must remain in $S_m$ for $T_m$ duration and then only the wake-up process can be initiated which takes $T_{sw}^{S_m}$ duration ($T_{sw}^{S_m}$- Sleep-to-wake-up time for sleep mode $S_m$). The OSMP-EO being an ONU-assisted protocol, in the worst possible case, an ONU may wake-up from sleep mode immediately after the arrival of its GATE message requiring a complete cycle to receive the next one when the desired bandwidth will be reported. The US transmission can only be initiated after receiving the next GATE message requiring another complete cycle. Thus, $ONU_i$ requires at most $2T_{cm}$ ($T_{cm}$- Maximum cycle time) duration to initiate the US transmission. 
 {Therefore, if $T_{mw}^{S_m}$ denotes the minimum value of the time interval between the time instants when the buffer fills up and the ONU initiates the wake-up process from sleep mode $S_m$ which is required to eliminate the possibility of packet drop, then $T_{mw}^{S_m}=T_{sw}^{S_m}+2T_{cm}+T_m$. }
 The rules for deciding $S_{nx}^i$ are the following:   
 \begin{itemize}
 	\item If $T_{bf,n}^i>T_{mw}^{S_m}$ then  $S_{nx}^i=S_m$ 
 	\item  Otherwise, it initiates the wake-up process from sleep mode $S_m$ ($S_{nx}^i=on$) 
 \end{itemize}
The OSMP-EO concentrates only on the worst case scenario in order to avoid the packet-drop possibility. However, in practical scenarios, the prediction error may cause a non-zero probability of packet-drop. This can easily be managed by setting a threshold ($N_{th}^i$), smaller than the actual buffer-size ($N_{sz}^i$). Hereafter, buffer-fill up time indicates the time taken to fill-up $N_{th}^i$ packets.
\subsubsection*{\textbf{Case} $\mathbf{S_c^i=on}$}
If $S_c^i=on$ then in OSMP-EO, $T_{nd}$ depends on the mode of $ONU_i$ in the previous decision instant, denoted as $S_{pr}^i$. If $S_{pr}^i=S_m\in\{ds,fs\}$ then the next decision is taken after up-streaming $N_{th}^i$ number of packets. On the other hand, in case of $S_{pr}^i=on$, the next decision is taken after up-streaming all packets that are currently stored in the buffer, denoted by $b_c^i$, which is same as our previous proposal \cite{chayan}.
We show in {Section} \ref{ssec:my} that due to the trade-off between the sleep-to-wake-up time and power consumption \cite{internetpowerconsumption}, a sleep mode provides improvement in energy-efficiency as compared to other modes with higher power consumption figures if $T_{bf,n}^i$ is longer than a threshold. This threshold for $ds$ and $fs$ is denoted as $T_{lb}^{ds}$ and $T_{lb}^{fs}$. The following conditions are followed for deciding $S_{nx}^i$: 
\begin{itemize}
	\item If $T_{bf,n}^i>T_{lb}^{ds}$ then select the  mode in the next decision point as deep sleep ($S_{nx}^i=ds$).
	\item If $T_{lb}^{fs}\leq T_{bf,n}^i<T_{lb}^{ds}$ then select the next mode as fast sleep ($S_{nx}^i=fs$).
	\item Otherwise, the active mode is retained ($S_{nx}^i=on$).
\end{itemize}

Now, we illustrate the OSMP-EO protocol with the aid of Fig. \ref{threshold}(a). Let, at an arbitrary decision time $t$, $ONU_i$ is in sleep mode $S_m$. Thus, the next decision is taken after $T_m$ duration by following the rules, discussed for the case: $S_c^i=S_m$. Let, at time $t_1$ ($=t+T_m$), the mode is decided as $on$. {Since, in the previous decision point ({i.e.} at $t$) mode was $S_m$, the next decision is taken after up-streaming $N_{th}^i$ number of packets which requires $T_{th}^i$ duration.} The mode, at $t_2$ ($=t_1+T_{th}^i$), is decided by following the rules, mentioned for the case of $S_{c}^i=on$. Let, at $t_2$, mode is again decided as $on$. As, in this case, $S_{pr}^i=on$, the mode will be decided after up-streaming $b_c^i$ number of packets (requires $T_{b_c^i}$ duration) by following the rules same as $t_1$.
\begin{figure*}[t]
	\centering
	\includegraphics[scale=0.49]{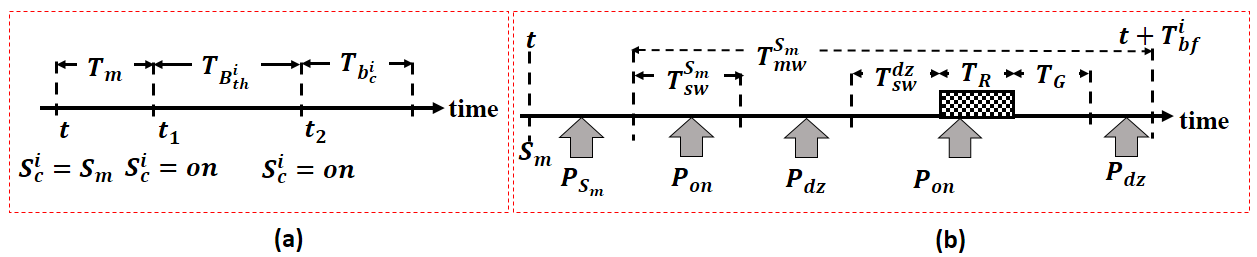}
	\caption{Illustration of (a) OSMP-EO protocol (b) energy consumption during sleep mode $S_m$. }
	\label{threshold}
\end{figure*} 
\subsection{Calculation of $T_{lb}^{ds}$ and $T_{lb}^{fs}$}\label{ssec:my}
Let, at the current decision instant, the buffer fill-up time is predicted as $T_{bf}^i$. We now explain  the energy consumption over this $T_{bf}^i$ time interval with the help of Fig. \ref{threshold}(b) if $ONU_i$ decides its mode as sleep mode $S_m$. In OSMP-EO, $ONU_i$ wakes up from sleep mode $S_m$,  $T_{sw}^{S_m}+2T_{cm}+T_m=T_{mw}^{S_m}$ duration before the buffer gets filled up. If the buffer fill-up time is correctly predicted then $ONU_i$ remains in sleep mode $S_m$ for $T_{bf}^i-T_{mw}^{S_m}$ duration when it consumes $P_{S_m}$ amount of power. Thereafter, $ONU_i$ requires $T_{sw}^{S_m}$ duration to wake-up from sleep mode $S_m$ when the power consumption is $P_{on}$. During the rest of the period ({i.e.} $2T_{cm}+T_m$) only one REPORT message will be transmitted. Since, in OSMP-EO, ONUs enter into doze mode immediately after waking up from sleep mode $S_m$, $ONU_i$ needs to wake-up from doze mode before sending the REPORT which takes $T_{sw}^{dz}$ duration. Thus, the power consumption over $T_{sw}^{dz}+T_R+T_G$ time interval is $P_{on}$ and during the remaining period power consumption is $P_{dz}$. Thus, the energy consumption of $ONU_i$ during sleep mode $S_m$ ($E_{S_m}$) is:
\begin{align*}
E_{S_m}=T_{bf}^iP_{S_m}+T_{sw}^{S_m}(P_{on}-P_{S_m})+(2T_{cm}+T_m)\\\times(P_{dz}-P_{S_m})+(T_R+T_G+T_{sw}^{dz})(P_{on}-P_{dz})
\end{align*}
 The achieved energy-efficiency of $ds$, is better than  $fs$ if $E_{ds}>E_{fs}$ and it requires:
\begin{align}\label{dsv}
\nonumber T_{bf}^{i}\geq\frac{T_{sw}^{fs}P_{fs}-T_{sw}^{ds}P_{ds}+(T_{sw}^{ds}-T_{sw}^{fs})P_{on}}{P_{fs}-P_{ds}}\\+2T_{cm}+T_m=T_{lb}^{ds}
 \end{align}
If $ONU_i$ remain active then energy consumption is $E_{on}=T_{bf}^iP_{on}^{avg,i}$. The fast sleep provides better energy savings as compared to active mode if $E_{fs}>E_{on}$ which requires:
\begin{align}\label{ls}
 \nonumber T_{bf}^{i}\geq\dfrac{T_{sw}^{fs}(P_{on}-P_{fs})+(2T_{cm}+T_m)(P_{dz}-P_{fs})}{P_{on}^{avg,i}-P_{fs}}\\+\dfrac{(T_R+T_G+T_{sw}^{dz})(P_{on}-P_{dz})}{P_{on}^{avg,i}-P_{fs}}=T_{lb}^{fs}
\end{align}
{ Next, we prove that energy-efficiency of the OSMP-EO protocol can be analyzed with the help of a Discrete Time Markov Chain (DTMC).} 
\section{Model formulation and description}\label{ssec:stated}
In this section, we formulate a Markov model of the proposed OSMP-EO protocol which is used for analyzing the achieved energy-efficiency figure. 
The OSMP-EO protocol, being an ONU-assisted protocol, the Markov analysis can be performed for each ONU separately. Without loss of generality, in this paper, we formulate the Markov model for $ONU_i$. Since, in this paper, we are focused on analyzing the performance of the OSMP-EO protocol, the grant-sizing scheme is taken to be the simplest one: the fixed grant-sizing scheme \cite{ipact}. We also explain how this analysis can be extended for other grant-sizing schemes (refer section \ref{ssec:statetr}). The following assumptions are taken in this analysis. 
\begin{itemize}
	\item The traffic arrival to $ONU_i$ follows Poisson process with mean $\lambda$ which has already been assumed in all existing queuing analysis in this area (EPON) \cite{downstreamcentric,multipowerlevel,sleepana,m1,m2,m3,m4}.
	\item The predicted buffer fill-up time always replicates the actual values. Thus, our analysis provides the best possible performance and it is independent of traffic prediction mechanism.  As the prediction mechanism is improved, the actual performance of the OSMP-EO protocol asymptotically follows our analytical results.  
\end{itemize}     
 
With these assumptions, we now try to formulate a Discrete Time Markov Model (DTMC) of OSMP-EO.
Let us define the mode of $ONU_i$ in the previous, current and the next observation points as $S_p^i$, $S_c^i$ and $S_n^i$. If the observation points are identical to the decision points then $S_{p}^i=S_{pr}^i$ and $S_{n}^i=S_{nx}^i$; otherwise they are different. Since, in each decision point (refer \ref{subsubsection:osmp}), $ONU_i$ decides its mode, an obvious choice for observation points of the DTMC are the decision instants. With this observation instant, we first try to find a state description that obeys the Memoryless property \cite{brts}, termed as valid state description. Firstly, we claim that both $S_c^i$ and $b_c^i$ must be included in the state description. 
\begin{claim}\thlabel{lemma 5}
	 The state description of DTMC must comprise of $S_c^i$ and $b_c^i$.
	\end{claim}
\begin{proof}
We know the energy consumption figures for different modes ($\{ds,fs,on\}$) are dissimilar. Thus, evaluation of energy-efficiency requires calculating probabilities of being in all modes, needing $S_c^i$ to be included into the state description. If $S_c^i$ individually hold the memoryless property then mode of $ONU_i$ of the next observation instant {i.e.} $S_{n}^i$ can be decided only by $S_c^i$. However, we now show that the decision about $S_n^i$ requires $b_c^i$ for the case of $S_c^i=S_m\in\{ds,fs\}$, which proves our claim. In case of $S_c^i=S_m$, $ONU_i$ decide $S_n^i$  by comparing $T_{bf,n}^i$ with $T_{mw}^{S_m}$ (refer {Section} \ref{subsubsection:osmp}).
$S_n^i$ is decided as $S_m$ if $T_{mw}^{S_m}<T_{bf,n}^i$ or in other words, $N_{th}^i>b_n^i+A_p^{i,T_{mw}^{S_m}}$ where $A_p^{i,T}$ is a random process that denotes the number of packets, predicted to arrive to $ONU_i$ over the time interval $T$. 
Thus, $S_n^i$ depends on $b_n^i$, and $b_n^i$ can be calculated as: $b_n^i=b_c^i+A^{i,T_{no}}$ where $T_{no}$ and $A^{i,T_{no}}$ denotes the time interval between the current and next observation instants and the traffic arrival over the time interval $T_{no}$ respectively. Hence, $S_n^i$ is dependent on $b_c^i$ necessitating $b_c^i$ to be included in the state description.
\end{proof}
 \thref{lemma 5} proves the necessity  of including both $S_c^i$ and $b_c^i$ into the state description. Our next step would then be to verify the sufficiency condition that is whether $S_c^i$ and $b_c^i$ jointly ($\{S_c^i,b_c^i\}$) hold the Markovian property so that $\{S_c^i,b_c^i\}$  qualifies as a valid state description for the DTMC. 
  Our assumption of ideal traffic prediction causes the predicted arrival process exactly same as actual one. In case of $S_c^i=S_m$, as $A^{i,T_{mw}^{S_m}}$ follows Poisson process, $A_p^{i,T_{mw}^{S_m}}$ is also Poisson distributed with mean $\lambda T_{mw}^{S_m}$. Since $T_{mw}^{S_m}$ depends only on $S_c^i=S_m$, $S_n^i$ dependents only on $S_c^i$ and $b_n^i$ (refer \thref{lemma 5}). Further, in this case, we know that $T_{no}=T_m$ (refer Section \ref{subsubsection:osmp}) and hence, $b_n^i$ can be calculated only by the information of $b_c^i$. Thus, $\{S_c^i, b_c^i\}$ seems to be a valid state description of the DTMC for the case $S_c^i=S_m$. Similar argument can further be provided for the case of $S_c^i=on$ as well. Thus, at the first glance, $\{S_c^i, b_c^i\}$ appears to be qualified as a valid state description with observation instants  exactly same as decision instants. However, a closer look into the protocol proves that $\{S_c^i, b_c^i\}$ violates Memoryless property. Not only that, we now claim that if the observation instances are identical to decision points then there doesn't exist any set of parameters at the current decision point that summarizes the complete history for the next instant which proves the non-Markovian nature of the model. 
  \begin{claim}\thlabel{cl2}
  	The modeling of OSMP-EO by choosing observation points identical to the decision instances is non-Markovian.
  \end{claim}
  \begin{proof}
  	At current time $t$, prediction is performed over $T_{pc}$ duration and next decision is taken after $T_{no}$ duration. Thus, there always exists an overlap region between $T_{pc}$ and $T_{no}$.
  	Owing to the assumption of ideal traffic prediction, $A_{p}^{i,T_{pc}}$ imposes an additional condition on deciding the next state ({i.e.} both $b_n^i$ and $S_n^i$). 
  	In OSMP-EO, if $S_p^i=S_m$ then $T_{pc}=T_{mw}^{S_m}$ while if $S_p^i=on$ then $T_{pc}=T_{lb}^{ds}$ or  $T_{lb}^{fs}$, depending on $S_c^i$ (refer Section \ref{subsubsection:osmp}) and therefore, $T_{pc}$ depends  on $S_p^i$. Hence, mode of $ONU_i$ in the next observation point ($S_n^i$) depends not only on $S_c^i$ but also on $S_p^i$. 
  	  \begin{figure*}[t]
  		\centering
  		\includegraphics[scale=0.48]{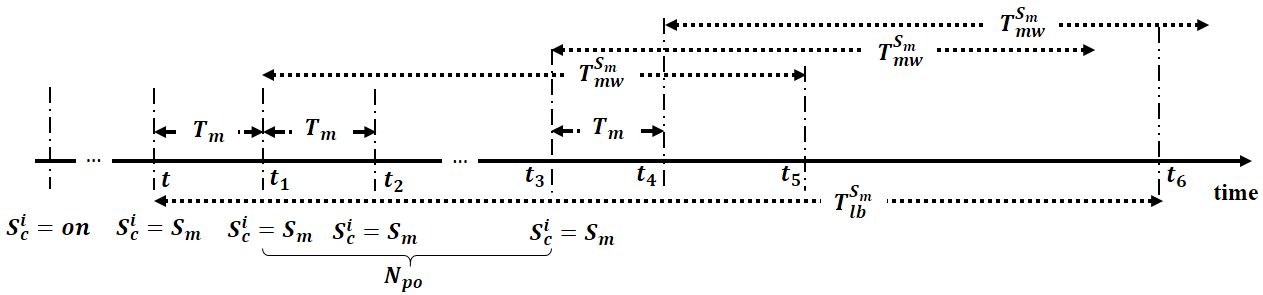}
  		\caption{Illustration of \thref{cl2}. }
  		\label{illast}
  	\end{figure*}
  	We now prove with the aid of Fig. \ref{illast} that $S_n^i$ depends not only on $S_p^i$ but also on the mode of multiple previous decision instants. Let, at time $t$, $ONU_i$ decides to change its state from active mode ($on$) to sleep mode $S_m$ (refer Fig. \ref{illast}). This requires satisfying the condition $N_{th}^i>b^i_c+A_p^{i,T_{lb}^{S_m}}$ (refer Section \ref{subsubsection:osmp}). Owing to the assumption of ideal traffic prediction, buffer of $ONU_i$ contains lesser than $N_{th}^i$ number of packets till $t_6$($=t+T_{lb}^{S_m}$). Since, at $t$, mode of $ONU_i$ is $S_m$, the next decision is taken after $T_{no}=T_m$ duration ({i.e.} at $t_1=t+T_m$)
  	 by comparing $N_{th}^i$  with $b^i_c+A_p^{i,T_{mw}^{S_m}}$ (refer {Section} \ref{subsubsection:osmp}). 
  	  As $T_{lb}^{S_m}>T_{mw}^{S_m}+T_m$ (refer eq. (\ref{dsv}) and eq. (\ref{ls})), the condition  $N_{th}^i>b^i_c+A_p^{i,T_{mw}^{S_m}}$ is surely be satisfied which ensures $S_c^i=S_m$ at $t_1$. By following the same argument, it can be proved that $S_c^i=S_m$ for the next $N_{po}=\big\lfloor\frac{T_{lb}^{S_m}-T_{mw}^{S_m}}{T_m}\big\rfloor$ number of observation instants (refer Fig. \ref{illast}).  
  	 Thus, $S_c^i$ of an observation instant depends on $S_c^i$ of the previous  $N_{po}$ number of points which proves the non-Markovian nature of OSMP-EO.      
  \end{proof}
\thref{cl2} proves that the proposed OSMP-EO protocol violates the Memoryless property.  We now reason that intelligent selection of state description and discrete observation points, different from the actual decision points, allows converting this  non-Markovian system to a Markovian one firstly for the case of $S_c^i=S_m\in\{ds,fs\}$ and then for $S_c^i=on$.
\subsubsection*{\textbf{Case} $\mathbf{S_c^i=S_m}$}
From the afore-mentioned discussion (refer \thref{cl2}), we understand that once $ONU_i$ enters into sleep mode $S_m$ from active mode, it retains the same state for $N_{po}$ observation instants which introduces the dependency on the past. This dependency will be absent if  $T_{no}>T_{lb}^{S_m}-T_{mw}^{S_m}$. Thus, we have to choose a value of $T_{no}$ such that $T_{no}>T_{lb}^{S_m}-T_{mw}^{S_m}$ is satisfied without loosing any information. Since, in OSMP-EO, $ONU_i$ decides its mode only at some discrete decision points (separated by $T_m$ in this case), observation instants must coincide with some decision points causing $T_{no}\geq(N_{po}+1)T_m$. At the time instant when $T_{no}=(N_{po}+1)T_m$, $ONU_i$ has a non-zero probability of entering into active mode. Thus, this information cannot be captured if $T_{no}>(N_{po}+1)T_m$. 
Consequently, $T_{no}=(N_{po}+1)T_m$ is the only possible option which is selected in the analysis for the case of $S_{p}^i=on$ and $S_c^i=S_m$. However, if $S_{p}^i=S_m$ and $S_c^i=S_m$ then the possibility of choosing $S_n^i=on$ is present in all decision points. Thus, in this case, we choose $T_{no}=T_m$. As a result, $T_{no}$ depends on $S_p^i$ resulting in violation of Markovian property. Nevertheless, the system can be converted to a Markovian process simply by including $S_p^i$ into the state description. It is now clear that $b_c^i$ and $S_c^i$ at the next observation point can be summarized by $b_c^i$, $S_c^i$ and $S_p^i$ of the current observation instant (refer {Section} \ref{ssec:statetr} for further discussion).  $S_p^i$ of the next state is identical to $S_c^i$ of the current state. Hence, $S_p^i$, $S_c^i$ and $b_c^i$ jointly holds Markovian property allowing $\{S_p^i,S_c^i,b_c^i\}$ as a valid state of the DTMC.

 \subsubsection*{\textbf{Case} $\mathbf{S_c^i=on}$}
 Our next step is to prove that $\{S_p^i,S_c^i,b^i_c\}$ is a valid state description even for $S_c^i=on$.   
If $S_p^i=S_m\in\{ds,fs\}$ and $S_c^i=on$, $ONU_i$ requires $T_{sw}^{S_m}$ duration to wake-up from sleep mode $S_m$. The OSMP-EO  being an ONU-assisted protocol, $ONU_i$ takes the decision of waking up from sleep mode $S_m$ without any prior information of the GATE message arrival time. Consequently, it is possible for $ONU_i$ to wake-up anywhere in between two consecutive GATE messages. The fixed grant sizing scheme ensures a constant time interval between two consecutive  GATE messages: $T_{cm}$. Thus, reception of next GATE message after waking up from sleep mode $S_m$ requires $T_{cm}/2$ duration on an average. Thereafter, the US transmission is initiated and it continues until $N_{th}^i$ number of packets get up-streamed which requires  $\big\lceil \frac{N_{th}^i}{N_m^i}\big\rceil T_{cm}$ time interval where $N_m^i$ is the maximum granted bandwidth of $ONU_i$ in one cycle. As a result, the next decision about its state is taken after $T_{sw}^{S_m}+\big\lceil \frac{N_{th}^i}{N_m^i}\big\rceil T_{cm}+0.5T_{cm}=T_{mo}^{S_m}$ duration. On the other hand, if $S_p^i=S_c^i=on$ then $ONU_i$ up-stream $b_i^c$ number of packets before taking further decision about $S_c^i$. Since, in this case, no time gets wasted for waking up, up-streaming $b_i^c$ number of packets requires $\big\lceil\frac{b_i^c}{N_m^i}\big\rceil T_{cm}$ duration which is the next decision point. In this case, if observation instants of the DTMC are considered to be same as decision points of the OSMP-EO protocol, by following the similar argument as provided for the case of $S_c^i=S_m$ (discussed above), it is evident that $\{S_p^i,S_c^i,b^i_c\}$ is a valid state description even for $S_c^i=on$. 
 {From the above discussion it is now clear that the OSMP-EO protocol can be modeled by selecting observation instances for different cases as discussed above, and the states as $\{S_p^i, S_c^i, b^i_c\}$.} The proposed DTMC of the OSMP-EO algorithm is presented in Fig. \ref{DTMC}(a) where all state transitions from states with $b^i_c=k$ is shown.         
In OSMP-EO, transition from one sleep mode to other is not possible, and therefore, the states $\{ds,fs,b^i_c~(\forall b^i_c)\}$ and $\{fs,ds,b^i_c~(\forall b^i_c)\}$ are unfeasible. $N_{sz}^i$ being the buffer size of $ONU_i$, states with $b_c^i>N_{sz}^i$ are also invalid. $ONU_i$ decides $S_c^i=S_m\in \{ds,fs\}$ only if the buffer of it contains lesser than $N_{th}^i$ number of packets till the next $T_{mw}^{S_m}$ duration and the next observation instant is after $T_m <T_{mw}^{S_m}$ duration. Thus, at the current as well as next observation instants, the buffer occupancy of $ONU_i$ is lesser than $N_{th}^i$. Consequently, the states $\{S_m,S_c^i,b_c^i\geq N_{th}^i\}$ and $\{on,S_m,b_c^i\geq N_{th}^i\}$ are invalid.  As $S_c^i$ of the current state becomes $S_p^i$ of the next state, state transitions from states with $S_c^i$ to states with $S_p^i\neq S_c^i$ are impossible. If in the current observation instant $S_c^i=S_m\in\{ds,fs\}$, the up-stream transmission of $ONU_i$ is terminated making state transition from states with $b^i_c=k$ to states with $b^i_c<k$ is impractical (refer Fig. \ref{DTMC}(a)). {From the above discussion it is now clear that the OSMP-EO protocol can be modeled by observing the system in some discrete time instants and the state space is finite. Hence, the proposed Markov model is a DTMC. } 
State transition probabilities for all other cases are calculated next.  
\section{Calculation of State Transition Probabilities} \label{ssec:statetr}  
 Here, state transition probabilities are calculated for all cases segregated based on all feasible combinations of  $S_p^i$ and $S_c^i$. 
 \subsection{Transition from states with $S_c^i=S_m\in\{ds,fs\}$}  
 Let, at the current observation time $t$, the state of $ONU_i$ be $\{S_p^i~(=S_m\text{ or } on),S_c^i=S_m,b_c^i=k\}$. At time $t$, for deciding $S_c^i=S_m$, $ONU_i$ must satisfy the condition $A_p^{i,T_{pc}}\leq N_{th}^i-k-1$ where $T_{pc}=T_{mw}^{S_m}$ if $S_p^i=S_m$ and otherwise, $T_{pc}=T_{lb}^{S_m}$ (refer {Section} \ref{subsubsection:osmp}). The next observation instant is after $T_{no}=T_m$ duration if $S_p^i=S_m$ while if $S_p^i=on$, $T_{no}=T_{no}=\Big(\lfloor\frac{T_{lb}^{S_m}-T_{mw}^{S_m}}{T_m}\rfloor+1\Big)T_m$ (refer {Section} \ref{ssec:stated}). As, during sleep mode, no US transmission is possible, at the next observation instant ({i.e.} at $t_1=t+T_{no}$), queue length will be $j(>k)$ ($b_c^i=j$) if during $T_{no}$ duration $j-k$ number of packets arrive ({i.e.} $A^{i,T_{no}}=j-k$). The rules, followed by $ONU_i$ for deciding $S_c^i$ at $t_1$, are:  $S_c^i=S_m$ if $A_p^{i,T_{pn}}\leq N_{th}^i-j-1$ and otherwise, $S_c^i=on$ where $T_{pn}=T_{mw}^{S_m}$
 Thus, at the next observation point the state of $ONU_i$ will be $\{S_p^i=S_m,S_c^i=S_m,b_c^i=j\}$ if  the events $A^{i,T_{no}}=j-k$ and $A_p^{i,T_{pn}}\leq N_{th}^i-j-1$ occurs simultaneously whereas, the state will be $\{S_p^i=S_m,S_c^i=on,b_i=j\}$ if both $A^{i,T_{no}}=j-k$ and $A_p^{i,T_{pn}}\geq N_{th}^i-j$ happens. However, due to the assumption of ideal prediction mechanism, the condition that is satisfied at time $t$ ({i.e.} $A_p^{i,T_{pc}}\leq N_{th}^i-k-1$) must hold true. Let the events $A^{i,T_{no}}=j-k$, $A_p^{i,T_{pn}}\leq N_{th}^i-j-1$, $A_p^{i,T_{pn}}\geq N_{th}^i-j$ and $A_p^{i,T_{pc}}\leq N_{th}^i-k-1$ is denoted by $Y_{j-k}$, $E_1$, $\xoverline[0.5]{E_1}$ and $E_2$ respectively.
 Probability of state transition from $\{S_p^i, S_m, k\}$ to  $\{S_m, S_m, j\}$ and $\{S_m, on, j\}$, denoted by $p_{S^i_p,S_m,k}^{S_m,S_m,j}$ and $p_{S_p^i,S_m,k}^{S_m,on,j}$ respectively, are given by:
  \begin{figure*}
  	\centering
 	\includegraphics[scale=.45]{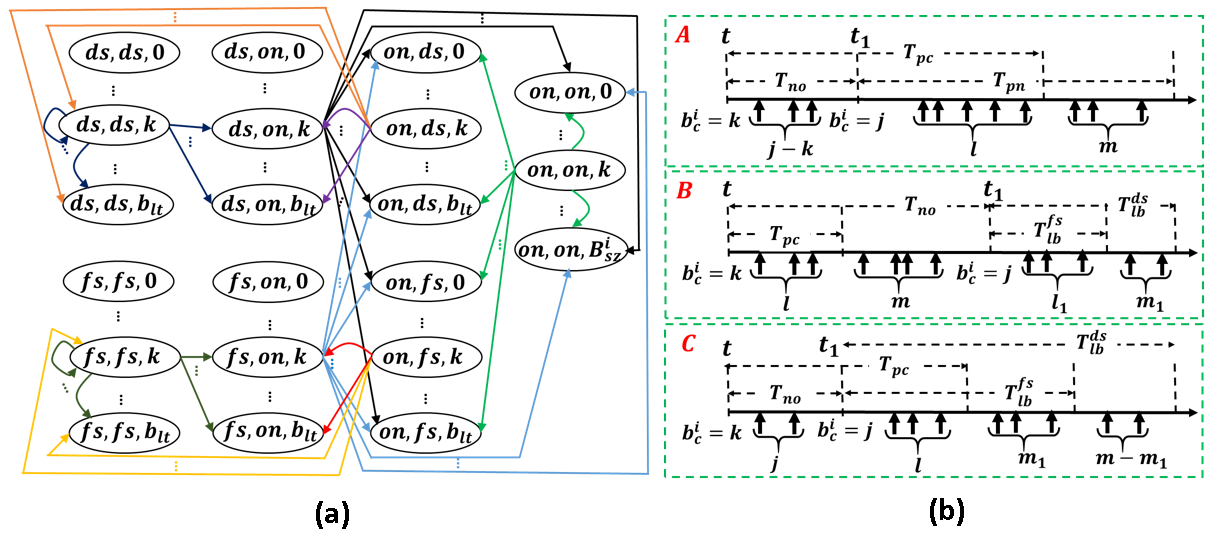}
 	\caption{(a) DTMC of OSMP-EO (b) Calculation of state transition probabilities. $b_{lt}=N_{th}^i-1$}
 	\label{DTMC}
 \end{figure*} 
\begin{align}\label{rwt}
 \nonumber p_{S_p^i,S_m,k}^{S_m,S_m,j}&=Pr\{Y_{j-k}, E_1 | E_2\}  ~~~~\text{and} \\ p_{S_p^i,S_m,k}^{S_m,on,j}&=Pr\{Y_{j-k}, \xoverline[0.5]{E_1} | E_2\}
\end{align}
where $Pr\{E\}$ denotes probability of occurrence of event $E$.
By using Bayes' theorem, eq. (\ref{rwt}) can be rewritten by eq. (\ref{num}).
\begin{align}\label{num}
\nonumber p_{S_p^i,S_m,k}^{S_m,S_m,j}&=\dfrac{Pr\{E_1,E_2 |Y_{j-k} \}}{Pr\{E_2\}} Pr\{Y_{j-k}\}~~~~\text{and} \\ p_{S_p^i,S_m,k}^{S_m,on,j}&=\dfrac{Pr\{\xoverline[0.5]{E_1},E_2 | Y_{j-k}\}}{Pr\{E_2\}}Pr\{Y_{j-k}\}
\end{align} 
 As the arrival process to $ONU_i$ follows the Poisson distribution with mean $\lambda$, $Pr\{Y_{j-k}\}$ is evaluated as the probability of arrival of $j-k$ packets from a Poisson traffic source with mean $\lambda T_{no}$, denoted by $\alpha_{j-k}^{\lambda T_{no}}$.
 Thus, $Pr\{Y_{j-k}\}$ is given by eq. (\ref{arr}).
 \begin{equation}\label{arr} 
Pr\{Y_{j-k}\}=\alpha_{j-k}^{\lambda T_{no}} \quad \textrm{where} \quad \alpha_l^\mu=\frac{e^{-\mu}{\mu}^{l}}{(l)!}
 \end{equation}	
We now calculate all other probability terms of eq. (\ref{num}) in \thref{l1}, \thref{l2} and \thref{l3}.
 \begin{lemma}\thlabel{l1}
 	$Pr\{E_2\}=\sum_{w=0}^{N_{th}^i-k-1}\alpha_{w}^{\lambda T_{pc}}$
 \end{lemma}
\begin{proof}
	Due to the assumption of ideal traffic prediction, the random process $A_p^{i,T_{pc}}$ become identical to $A^{i,T_{pc}}$ which follows Poisson distribution with mean $\lambda T_{pc}$. Thus, the event $E_2$ occurs if this Poisson process generates lesser than $N_{th}^i-k$ packets which proves this lemma. 
\end{proof}

\begin{lemma}\thlabel{l2}
	If $T_1=T_{pc}-T_{no}$ and $T_2=T_{no}+T_{pn}-T_{pc}$ then $Pr\{E_1,E_2 |Y_{j-k} \}=\sum_{l=0}^{N_{th}^i-j-1}\sum_{m=0}^{N_{th}^i-j-l-1}\alpha_l^{\lambda T_1}\alpha_m^{\lambda T_2}$.
\end{lemma}
\begin{proof}
	Let $l$ and $m$ number of packets arrive over the time interval $T_1$ and $T_2$ respectively as shown in A of Fig \ref{DTMC}(b). Since, at $t$, $b_c^i=k$, the event $Y_{j-k}$ results in $b_c^i=j$ at $t_1$ and hence,  $A_p^{i,T_{pc}}=j-k+l$ and $A_p^{i,T_{pn}}=l+m$.   
	  The event $E_2$ and $E_1$ restricts $l$ and $m$ by $N_{th}^i-j-1$ and  $N_{th}^i-j-l-1$ respectively. Hence, $Pr\{E_1,E_2 |Y_{j-k} \}$ is evaluated as the probability of arrival of $l$ and $m$ number of packets over the time interval $T_1$ and $T_2$ by a Poisson source with mean arrival rate $\lambda$ where $l \leq N_{th}^i-j-1$ and $m\leq N_{th}^i-j-l-1$. This proves \thref{l3}.
\end{proof}
\begin{lemma}\thlabel{l3}
	If $T_1=T_{pc}-T_{no}$ and $T_2=T_{no}+T_{pn}-T_{pc}$ then $Pr\{\xoverline[0.5]{E_1},E_2 | Y_{j-k}\}=\sum_{l=0}^{N_{th}^i-j-1}\sum_{m=N_{th}^i-j-l}^{\infty}\alpha_l^{\lambda T_1}\alpha_m^{\lambda T_2}$.
\end{lemma}
\begin{proof}
	\thref{l3} can be easily proved by following similar argument as \thref{l2}.
\end{proof}
\subsection{Transition from states with $S_p^i=S_m \in \{ds, fs\}$) and $S_c^i=on$}\label{ssec:Smon}
Let, at current observation instant $t$, the state of $ONU_i$ is $\{S_m, on,k\}$. For changing the mode from $S_m$ to $on$, at $t$, $ONU_i$ must ensures $A_{p}^{i,T_{pc}}\geq N_{th}^i-k$ where $T_{pc}=T_{mw}^{S_m}$. The next observation instant is after $T_{no}=T_{sw}^{S_m}+\big\lceil \frac{N_{th}^i}{N_m^i}\big\rceil T_{cm}+\frac{T_{cm}}{2}=T_{mo}^{S_m}$ duration ({i.e.} at $t_1=t+T_{no}$) when $N_{th}^i$ number of packets depart. Thus, at $t_1$, $b_c^i=j$ requires $A^{i,T_{no}}=j-k+N_{th}^i$. An approximate analysis of OSMP-EO protocol  for grant sizing schemes different from fixed one can be carried out in the similar manner only by replacing  $T_{no}=T_{sw}^{S_m}+\big\lceil \frac{N_{th}^i}{N_m^i}\big\rceil T_c^{avg}+0.5T_c^{avg}$ where $T_{c}^{avg}$ is the average cycle time of the OLT, calculated by analyzing the underlying Dynamic Bandwidth Allocation (DBA) scheme. 
At $t_1$, $S_c^i=ds$ if $A_{p}^{i,T_{lb}^{ds}}\leq N_{th}^i-j$ (refer {Section} \ref{subsubsection:osmp}). We denote the events: $A_{p}^{i,T_{pc}}\geq N_{th}^i-k$, $A_{p}^{i,T_{lb}^{ds}}< N_{th}^i-j$ by $\xoverline[.5]{E_2}$, $E_{ds}$ respectively. Hence, the transition probability from state $\{S_m, on,k\}$ to $\{on,ds,j\}$, denoted by $p_{S_m,on,k}^{on,ds,j}$, is given by eq. (\ref{starts}).
\begin{align}\label{starts}
p_{S_m,on,k}^{on,ds,j}&=Pr\{Y_{j-k+N_{th}^i},E_{ds}|\xoverline[.5]{E_2}\}
\end{align}
Since, in this case, $T_{pc}<T_{no}$, the time interval over which the prediction is perform to decide $S_c^i$ at the next observation instant ({i.e.} $T_{pn}$) does not intersects with both $T_{pc}$ and $T_{no}$ (refer B of Fig. \ref{DTMC}(b)). Thus, $E_{ds}$ is independent of both $Y_{j-k+N_{th}^i}$ and $\xoverline[.5]{E_2}$. By using  Bayes' Theorem and then this independence condition, eq. (\ref{starts}) turns out to be eq. (\ref{dk}).
 \begin{equation}\label{dk}
p_{S_c,on,k}^{on,ds,j}=\frac{Pr\{ Y_{j-k+N_{th}^i},\xoverline[.5]{E_2}\}}{Pr\{\xoverline[.5]{E_2}\}}Pr\{E_{ds}\}
\end{equation}
We know, if $A_{p}^{i,T_{lb}^{fs}}<N_{th}^i-j$ and $A_{p}^{i,T_{lb}^{ds}}\geq N_{th}^i-j$ occur simultaneously then at $t_1$, $S_c^i=fs$ and otherwise, $S_c^i=on$. Let the events:  $A_{p}^{i,T_{lb}^{ds}}\geq N_{th}^i-j$,  
$A_{p}^{i,T_{lb}^{fs}}< N_{th}^i-j$ and $A_{p}^{i,T_{lb}^{fs}}\geq N_{th}^i-j$ is denoted by $\xoverline[.5]{E_{ds}}$, $E_{fs}$, $\xoverline[.5]{E_{fs}}$ respectively. Thus, transition probabilities from state $\{S_m, on,k\}$ to  $\{on,fs,j\}$ and $\{on,on,j\}$, denoted by  $p_{S_m,on,k}^{on,fs,j}$ and $p_{S_m,on,k}^{on,on,j}$ respectively, are given by eq. (\ref{start1}).
 \begin{align}\label{start1}
\nonumber p_{S_c,on,k}^{on,fs,j}&=\frac{Pr\{ Y_{j-k+N_{th}^i},\xoverline[.5]{E_2}\}}{Pr\{\xoverline[.5]{E_2}\}}Pr\{\xoverline[.5]{E_{ds}}E_{fs}\}
~~~~\text{and}\\p_{S_c,on,k}^{on,on,j}&=\frac{Pr\{ Y_{j-k+N_{th}^i},\xoverline[.5]{E_2}\}}{Pr\{\xoverline[.5]{E_2}\}}Pr\{\xoverline[.5]{E_{fs}}\}
\end{align}
  It is evident that if $j\geq N_{th}^i$ then $Pr\{{E_{ds}}\}=Pr\{\xoverline[.5]{E_{ds}}E_{fs}\}=0$ and $Pr\{\xoverline[.5]{E_{fs}}\}=1$. Since $N_{sz}^i$ is the buffer size, if $A^{i,T_{no}}> N_{sz}^i+N_{th}^i-k$ then the buffer remains full ({i.e.} $b_c^i=N_{sz}^i$) and the excess packets will drop causing state transition from state with $b_c^i=k$ to state with $b_c^i=N_{sz}^i$. 
  Hence, $p_{S_m,on,k}^{on,on,N_{sz}^i}$  is given by eq. (\ref{onB}).
\begin{align}\label{onB}
\nonumber p_{S_m,on,k}^{on,on,N_{sz}^i}=\sum_{j=N_{sz}^i}^{\infty}{Pr\{ Y_{j-k+N_{th}^i}|\xoverline[.5]{E_2}\}}\\=\dfrac{\sum_{j=N_{sz}^i}^{\infty}{Pr\{ Y_{j-k+N_{th}^i},\xoverline[.5]{E_2}\}}}{Pr\{\xoverline[.5]{E_2}\}}
\end{align}
  
Now, we calculate all probability terms of eq. (\ref{dk}), eq. (\ref{start1}) and eq. (\ref{onB}) in \thref{l4}\textendash\thref{l7}. 
\begin{lemma}\thlabel{l4}
	$Pr\{\xoverline[.5]{E_2}\}=\sum_{w=N_{th}^i-k}^{\infty}\alpha_w^{\lambda {T_{mw}^{S_m}}}$
\end{lemma} 
\begin{proof}
		We know, $Pr\{\xoverline[.5]{E_2}\}=1-Pr\{E_2\}$ where $Pr\{E_2\}$ is given by \thref{l1}.
\end{proof}
\begin{lemma}\thlabel{l5}
	$Pr\{E_{S_m}\}=\sum_{w=0}^{N_{th}^i-j-1}\alpha_w^{\lambda {T_{lb}^{S_m}}}$ and 	$Pr\{\xoverline[.5]{E_{S_m}}\}=\sum_{w=N_{th}^i-j}^{\infty}\alpha_w^{\lambda {T_{lb}^{S_m}}}$, $S_m\in \{ds,fs\}$
\end{lemma}
\begin{proof}
	$Pr\{\xoverline[.5]{E_{S_m}}\}=1-Pr\{E_{S_m}\}$ where $Pr\{E_{S_m}\}$ is calculated similar to \thref{l1}.
\end{proof}
\begin{lemma}\thlabel{l6}
	$Pr\{ Y_{j-k+N_{th}^i},\xoverline[.5]{E_2}\}=\sum_{l=N_{th}^i-k}^{N_{th}^i-k+j}\alpha_l^{\lambda T_{pc}}\alpha_{j-k-l+N_{th}^i}^{\lambda (T_{no}-T_{pc})}$ where $T_{pc}=T_{mw}^{S_m}$.
\end{lemma}
\begin{proof}
	 Let $l$ and $m$ number of packets arrive over $T_{pc}$ and next $T_{no}-T_{pc}$ duration respectively (refer B of Fig. \ref{DTMC}(b)) resulting in $A^{i,T_{no}}=l+m$. The event $\xoverline[.5]{E_2}$ imposes a lower limit of $l$: $l\geq N_{th}^i-k$. Further, the event $Y_{j-k+N_{th}^i}$ ensures $m=j-k-l+N_{th}^i$. $m\geq 0$ restricts the value of $l$ by; $l\leq N_{th}^i-k+j$. This proves \thref{l6}. 
\end{proof}
\begin{lemma}\thlabel{l7}
	$Pr\{\xoverline[.5]{E_{ds}}E_{fs}\}=\sum_{l_1=0}^{N_{th}^i-j-1}\sum_{m_1=N_{th}^i-j-l_1}^{\infty}\alpha_{l_1}^{\lambda T_{lb}^{fs}}\alpha_{m_1}^{\lambda (T_{lb}^{ds}-T_{lb}^{fs})}$
\end{lemma}
\begin{proof}
	Let $l_1$ and $m_1$ number of packets arrive over time interval $T_{lb}^{fs}$ and next  $T_{lb}^{ds}-T_{lb}^{fs}$ (refer B of Fig. \ref{DTMC}(b)) causing  $A_p^{i,T_{lb}^{ds}}=l_1+m_1$. The events  $E_{fs}$ and $\xoverline[.5]{E_{ds}}$ restricts $l_1$ and $m_1$ by; $l_1\leq N_{th}^i-j-1$ and $m_1\geq N_{th}^i-j-l_1$ respectively which proves \thref{l7}.
\end{proof}
\subsection{Transition from states with $S_p^i=S_c^i=on$}
Let, at current time $t$, $ONU_i$ retain the active state which requires satisfying the condition $A_p^{i,T_{lb}^{fs}}\geq N_{th}^i-j$ ({i.e.} $\xoverline[.5]{E_{2}}$) if $b_c^i<N_{th}^i$ whereas if $b_c^i\geq N_{th}^i$ then no condition is needed to be satisfied. In this case, $T_{no}=\lceil b_c^i/N^i_{m} \rceil T_{cm}$ (refer {Section} \ref{ssec:stated}) when all $b_c^i=k$ number of packets get up-streamed. Thus, at $t_1=t+T_{no}$, $b_c^i=j$ requires $A^{i,T_{no}}=j$. Extension of this analysis for other grant sizing schemes can be performed by modifying $T_{no}$ in the similar manner as discussed above. Further, at $t_1$, $S_c^i$ is decided by verifying the following conditions: $S_c^i=ds$ if $E_{ds}$ occurs, while if both $\xoverline[.5]{E_{ds}}$ and $E_{fs}$ happens then $S_c^i=fs$ and otherwise, $S_c^i=on$.    
 Now, we calculate the state transition probability for three different cases: $T_{no}<T_{lb}^{fs}, b_c^i<N_{th}^i$; $T_{no}>T_{lb}^{fs}, b_c^i<N_{th}^i$; $b_c^i\geq N_{th}^i$.
\subsection*{\textbf{Case 1 (}$\mathbf{T_{no}<T_{lb}^{fs}, b_c^i=k<N_{th}^i}$\textbf{)}}
By following similar steps as used for calculating eq. (\ref{num}), 
$p_{on,on,k}^{on,ds,j}$, $p_{on,on,k}^{on,fs,j}$, and $p_{on,on,k}^{on,on,j}$ is given by eq. (\ref{ondsc}).
\begin{align}
\nonumber p_{on,on,k}^{on,ds,j}&=\frac{Pr\{\xoverline[.5]{E_2},E_{ds}|Y_j\}}{Pr\{\xoverline[.5]{E_2}\}}\alpha_{j}^{\lambda T_{no}}\label{ondsc},\\\nonumber p_{on,on,k}^{on,fs,j}&=\frac{Pr\{\xoverline[.5]{E_2},\xoverline[.5]{E_{ds}}E_{fs}|Y_{j}\}}{Pr\{\xoverline[.5]{E_2}\}}{\alpha_{j}^{\lambda T_{no}}},\\ p_{on,on,k}^{on,on,j}&=\frac{Pr\{\xoverline[.5]{E_2},\xoverline[.5]{E_{fs}}|Y_j\}}{Pr\{\xoverline[.5]{E_2}\}}\alpha_{j}^{\lambda T_{no}}
\end{align}
$Pr\{\xoverline[.5]{E_2}\}$ is calculated by using \thref{l4}. 
All other probability terms of eq. (\ref{ondsc}) are calculated by using \thref{l8}\textendash\thref{l10}. For proving  \thref{l8}\textendash\thref{l10}, we assume that $l$, $m_1$ and $m-m_1$ number of packets arrives over three consecutive time intervals $T_1=T_{pc}-T_{no}$, $T_2=T_{no}+T_{lb}^{fs}-T_{pc}$ and $T_3=T_{lb}^{ds}-T_{lb}^{fs}$ respectively as shown in C of Fig. \ref{DTMC}(b). The event $Y_j$ ensures  $A^{i,T_{pc}}=j+l$, $A_p^{i,T_{lb}^{fs}}=l+m_1$ and $A_p^{i,T_{lb}^{ds}}=l+m$.  
\begin{lemma}\thlabel{l8}
	$Pr\{\xoverline[.5]{E_2},E_{ds}|Y_j\}=\sum_{l=\max(0,N_{th}^i-k-j)}^{N_{th}^i-j-1}\sum_{m=0}^{N_{th}^i-j-l-1}\alpha_{l}^{\lambda T_1}\alpha_{m}^{\lambda (T_2+T_3)}$
\end{lemma}
\begin{proof}
	 The events $\xoverline[.5]{E_2}$ and $E_{ds}$ requires $l\geq N_{th}^i-k-j$ and $m\leq N_{th}^i-j-l-1$. We know both $l$ and $m$ are positive. $m>0$ upper bounds the values of $l$ by $N_{th}^i-j-1$. Thus, the event $\xoverline[.5]{E_2},E_{ds}|Y_j$ requires arrival of $l$ and $m$ number of packets over the time interval $T_1$ and $T_2+T_3$ with all restrictions on $l$ and $m$ as discussed above which proves \thref{l8}.
\end{proof}
\begin{lemma}\thlabel{l9}
	$Pr\{\xoverline[.5]{E_2},\xoverline[.5]{E_{ds}}E_{fs}|Y_j\}$\begin{align*}
	=\sum_{l=\max(0,N_{th}^i-k-j)}^{N_{th}^i-j-1}\sum_{m_1=0}^{N_{th}^i-j-l-1}\sum_{m=N_{th}^i-j-l}^{\infty}\alpha_{l}^{\lambda T_1}\alpha_{m_1}^{\lambda T_2} \alpha_{m-m_1}^{\lambda T_3}
	\end{align*}
\end{lemma}
\begin{proof}
	The events $\xoverline[.5]{E_2}$, $E_{fs}$ and $\xoverline[.5]{E_{ds}}$ requires $l\geq N_{th}^i-k-j$, $m_1\leq N_{th}^i-j-l-1$ and $m\geq N_{th}^i-j-l$ respectively. Thus, for all $m$ and $m_1$, $m\geq m_1$. Further, both $l$ and $m_1$ are non-negative. $m_1\geq 0$ necessitates $l\leq N_{th}^i-j-1$. 
	Thus, the event $\xoverline[.5]{E_2},\xoverline[.5]{E_{ds}}E_{fs}|Y_j$ occurs if $l$, $m_1$ and $m-m_1$ number of packets over the time interval $T_1$, $T_2$ and $T_3$ respectively with all restrictions on $l$, $m_1$ and $m$ as discussed above which proves \thref{l9}. 
\end{proof}
\begin{lemma}\thlabel{l10}
	$Pr\{\xoverline[.5]{E_2},\xoverline[.5]{E_{fs}}|Y_j\}=\sum_{l=\max(0,N_{th}^i-k-j)}^{\infty}\sum_{m_1=\max(0,N_{th}^i-l-j)}^{\infty}\alpha_{l}^{\lambda T_1}\alpha_{m_1}^{\lambda T_2} $
\end{lemma}
\begin{proof}
	The events $\xoverline[.5]{E_2}$ and $\xoverline[.5]{E_{fs}}$ ensures $l\geq N_{th}^i-k-j$ and $m_1\geq N_{th}^i-l-j$ respectively. Further, both $l$ and $m_1$ are positive. By using this, \thref{l10} can be easily proved.
\end{proof}
\subsubsection*{\textbf{Case 2 (}$\mathbf{t_{no}\geq t_{pp}, b_c^i<N_{th}^i}$\textbf{)}}
 This case is exactly same to the case $S_p^i=S_m \in \{ds,fs\}$, $S_c^i=on$ (refer {Section} \ref{ssec:Smon}).
  \subsubsection*{\textbf{Case 3 (}$\mathbf{b_c^i=k\geq N_{th}^i}$\textbf{)}}
  In this case, $ONU_i$ enters into the ON state without any prediction. Further, it can be easily noticed that in this case, the arrival process over $T_{no}$ duration is independent of the arrival process over $T_{lb}^{ds}$ or $T_{lb}^{fs}$. Hence, in this case, state transition probabilities $p_{on,on,k}^{on,ds,j}$, $p_{on,on,k}^{on,fs,j}$, $p_{on,on,k}^{on,on,j}$ ($j<N_{sz}^i$), and $p_{on,on,k}^{on,on,N_{sz}^i}$ is given by \ref{ky}.
  \begin{align}\label{ky}
  \nonumber
  p_{on,on,k}^{on,ds,j}&=\alpha_{j}^{\lambda T_{no}} \sum_{l=0}^{N_{th}^i-j-1}\alpha_l^{\lambda T_{lb}^{ds}},\\\nonumber
 p_{on,on,k}^{on,fs,j}&=\alpha_{j}^{\lambda T_{no}} \sum_{l=0}^{N_{th}^i-j-1}\alpha_l^{\lambda T_{lb}^{fs}}\sum_{m=N_{th}^i-j-l}^{\infty}\alpha_m^{\lambda T_3},\\\nonumber
 p_{on,on,k}^{on,on,j}&=\alpha_{j}^{\lambda T_{no}} \sum_{l=\max (0,N_{th}^i-j)}^{\infty}\alpha_l^{\lambda T_{lb}^{fs}},\\
p_{on,on,k}^{on,ds,N_{sz}^i}&=\sum_{l=N_{sz}^i}^{\infty}\alpha_{l}^{\lambda T_{no}}
\end{align}
    \begin{figure*}[t]
	\centering
	\includegraphics[scale=.55]{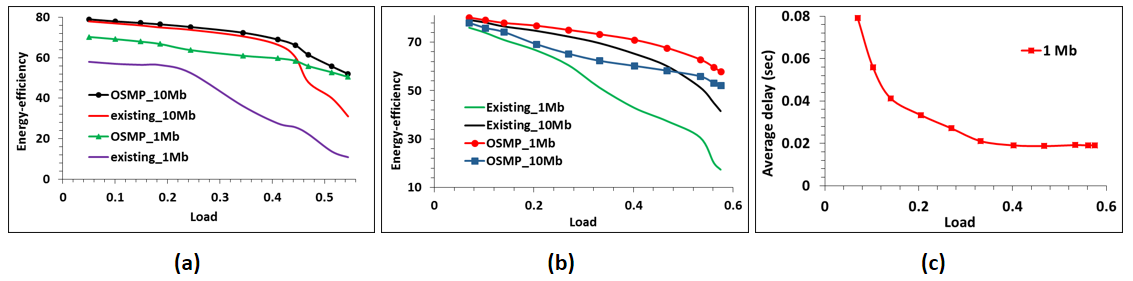}
	\caption{Comparison with existing protocols for scenarios: (a) only one ONU (say $ONU_i$) (b) all ONUs follows OSMP-EO protocol, (c) delay characteristic of OSMP-EO.}
	\label{fig:comp}
\end{figure*}
 \section{Calculation of State Probabilities} 
 Here, we calculate all state probabilities. In OSMP-EO, transition between two sleep modes is not possible. Further, $S_p^i$ of the next state is exactly same as $S_c^i$ of the current state (refer {Section} \ref{ssec:stated}). If $S_c^i=S_m \in\{ds,fs\}$ then the up-stream transmission gets terminated, eliminating the possibility of state transition between states with $b_c^i=k$ to states with $b_c^i<k$. By considering these restrictions on state transitions,    
 the probabilities of being in states $\{S_m, S_c^i , j\}$ ($\pi_{S_m,S_c^i,j}$)  and $\{on, S_c^i,j\}$ ($\pi_{on,S_c^i,j}$) are given by eq. (\ref{key1}) and eq. (\ref{key2}) respectively (refer Fig. \ref{DTMC}(a)).
 \begin{align}\label{key1} 
 \pi_{S_m,S_c^i,j}=\sum_{k=0}^{j}\pi_{S_m,S_m,k} p_{S_m,S_m,k}^{S_m,S_c^i,j}+\sum_{k=0}^{j}\pi_{on,S_m,k} p_{on,S_m,k}^{S_m,S_c^i,j}\\\nonumber\label{key2}\pi_{on, S_c^i,j}=\sum_{k=0}^{N_{th}^i-1}\pi_{ds,on,j} p_{ds,on,j}^{on, S_c^i,j}+\sum_{k=0}^{N_{th}^i-1}\pi_{fs,on,j} p_{fs,on,j}^{on, S_c^i,j}\\+\sum_{k=0}^{N_{sz}^i}\pi_{on,on,j} p_{on,on,j}^{on, S_c^i,j}
 \end{align}
 \section{Calculation of average energy-efficiency}
 In this {Section}, we analyze the average energy savings figure of the OSMP-EO protocol. $ONU_i$ retains the mode of the current observation instant ($S_c^i$) till the next observation instant {i.e.} for a duration of $T_{no}$. We denote $T_{no}$ and the energy consumption of $ONU_i$ over $T_{no}$ for state $\{S_p^i,S_c^i,b_c^i\}$ as $T^{S_p^i,S_c^i,b_c^i}$ and $E^{S_p^i,S_c^i,b_c^i}$ respectively. If $S_c^i=S_m\in\{ds,fs\}$ then $T_{no}=T_m$ and power consumption is $P_{S_m}$ (refer Section \ref{ssec:stated}) and hence, $E^{S_p^i,S_m,b_c^i}=P_{S_m}T_m$.
In case of $S_p^i=S_m$ and $S_c^i=on$, $T_{no}=T^{S_m,on,b_c^i}=T_{sw}^{S_m}+\big\lceil \frac{N_{th}^i}{N_m^i}\big\rceil T_{cm}+\frac{T_{cm}}{2}$. Among this  period $ONU_i$ takes $T_{sw}^{S_m}$ duration to wake-up from sleep mode when power consumption is $P_{on}$. After waking up, $ONU_i$ requires $1.5T_{cm}$ duration to start the up-stream data transmission when only one REPORT message is up-streamed. The rest of the period ($t_{rp}=\lceil N_{th}^i/N^i_{m} \rceil-1$) is used for up-stream data transmission when the power consumption is $P_{on}^{avg,i}$ (refer eq. (\ref{avgac})). Thus, in this case, $E^{S_m,on,b_c^i}=T_{sw}^{S_m}P_{on}+(1.5T_{cm}-T_{sw}^{dz}-T_R-T_G)P_{dz}+(T_{sw}^{dz}-T_R-T_G)P_{on}+t_{rp}P_{on}^{avg,i}$. In case of $S_p^i=S_c^i=on$, over entire $T_{no}=T^{on,on,b_c^i}=\big\lceil\frac{b_i^c}{N_m^i}\big\rceil T_{cm}$ duration power consumption is $P_{on}^{avg,i}$ and hence, $E^{on,on,b_c^i}=T_{no}P_{on}^{avg,i}$. 
Therefore, the average power consumption ($P_{avg}$) of the OSMP-EO algorithm is given by eq. (\ref{ty}).
 \begin{align}\label{ty}
 P_{avg}=\frac{\sum_{S_p^i}\sum_{S_c^i}\sum_{j}E^{S_p^i,S_c^i,j}\pi_{S_p^i,S_c^i,j}}{\sum_{S_p^i}\sum_{S_c^i}\sum_{j}T^{S_p^i,S_c^i,j}\pi_{S_p^i,S_c^i,j}}
 \end{align}
 However, if no energy-efficient mode is employed then  power consumption is always $P_{on}$. Hence, the average energy-efficiency of OSMP-EO algorithm is given by: $\eta_{avg}=1-\dfrac{P_{avg}}{P_{on}}$.
  \section{Result and discussion}
  In this {Section}, firstly, the average energy-efficiency figures of the OSMP-EO protocol are compared with the same of  other existing schemes. We then validate the analytical results with the results obtained from simulations. We next compare the results obtained by using our analytical model with that generated by simulations for self-similar traffic. {All simulations are performed in OMNET++ for a network runtime of $50s$ and the results are plotted with $95\%$  confidence interval.} 
  The link rate of the feeder fiber, the maximum traffic arrival rate at each ONU, the packet size and $T_m$ are assumed to be $1Gbps$, $100Mbps$, $1500Bytes$ and $0.5ms$ respectively \cite{chayan}. {Sleep-to-wake-up time of $ds$, $fs$, and $dz$ are considered to be $5.125ms$, $125\mu s$ and $1\mu s$ respectively \cite{chayan}. Whereas power consumption of $ds$, $fs$, $dz$ and $on$ are $0.75W$, $1.28W$, $2.39W$ and $3.984W$ respectively \cite{chayan}. } 
\begin{figure*}[t]
	\centering
	\includegraphics[scale=.6]{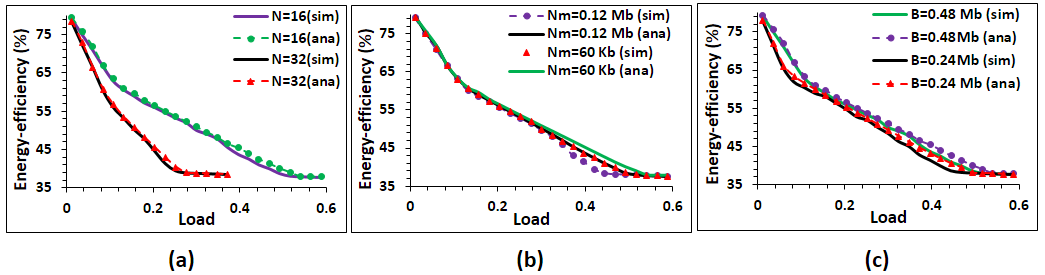}
	\caption{Validation of analytical model for different values of (a) $N$ (b) $N_m^i$ (c) $N_{th}^i$. Nm-$N_m^i$, B-$N_{th}^i$}
	\label{fig:test}
\end{figure*}  
\subsection{Comparison with existing literature}\label{sec:compari}
  {We have already demonstrated in \cite{chayan} that our previous proposal outperforms all existing protocols. Here, we compare the OSMP-EO protocol with our previous proposal where no sleep mode was employed during active periods.} 
  In OSMP-EO, we introduce a mechanism of employing doze mode during the active periods which provides additional energy savings. {In order to quantify this improvement, in Fig. \ref{fig:comp}(a) and Fig. \ref{fig:comp}(b), the average energy-efficiency figures are plotted as a function of traffic load for $N_{th}^i=N_{sz}^i=1Mb ~\text{and } 10Mb$ for two different sceneries: (i) only one ONU (say $ONU_i$) follows OSMP-EO protocol while all other ONUs are traditional ONUs and (ii) all ONUs follows OSMP-EO protocol respectively. The average delay of OSMP-EO protocol for $N_{th}^i=N_{sz}^i=1$ Mb is plotted if Fig. \ref{fig:comp}(c).} Here, traffic load of an ONU is defined with respect to the maximum arrival rate to that ONU. The traffic arrivals are considered to be self-similar with Hurst parameter ($H$) $0.8$ and it is generated by aggregating $16$ ON-OFF Pareto sources. Traffic prediction is performed by using ARMA(2,2) model in a similar manner as discussed in our previous proposal \cite{chayan}. {Fig. \ref{fig:comp}(a) and Fig. \ref{fig:comp}(b) depict that at low load, especially for the case of $N_{th}^i=10$ Mb, the improvement in energy efficiency,  as compared to our previous proposal \cite{chayan} is very small whereas at high load a significant energy benefit ($\sim 40\%$ for $N_{th}^i=1Mb$) can be achieved. This is due to the following fact. Both the reduction of traffic load and the increment of $N_{th}^i$ enhances the buffer fill-up time. As a result, the sleep duration increases and hence, ONUs wake-up from sleep mode less frequently. In OSMP-EO, a certain time interval is wasted when an ONU wakes-up from sleep mode and this wastage can be reduced by increasing of $N_{th}^i$ or by reducing traffic load. Further, the reduction of traffic load decreases the time interval over which an ONU up-streams. As a result of these two facts, the total active period gets reduced with the reduction of traffic load and increment of $N_{th}^i$.  Since OSMP-EO enhances energy efficiency by introducing doze mode during active periods; at a low load, for s higher value of $N_{th}^i$ (when the active period is very small), the increment is negligible. Whereas, at a high load with low value of $N_{th}^i$, the improvement is significant.}  Further, since the increment of traffic load reduces the sleep duration, an increase of traffic load results in a decrement of average delay as seen in \ref{fig:comp}(c).     

\subsection{Validation of analytical model with simulations}  \label{sec:validation} 
Here, the proposed Markov model of the OSMP-EO protocol is validated with simulation results for Poisson traffic. It is well known that in case of Poisson traffic, the current arrival is independent of the previous arrivals, and hence, traffic prediction with the knowledge of previous arrivals is impossible. In  simulations, we consider that $ONU_i$ predicts the number of packet arrivals over $T_{od}$ duration by $\lambda_i T_{od}$ where $\lambda_i$ is average packet arrival rate to $ONU_i$. 
 We assume the buffer size $N_{sz}^i$ as $1.2Mb$ while all other parameters are same as they are taken in Section \ref{sec:compari}.
The achieved energy-efficiency figures as a function of traffic load along with the same obtained from our analytical model  for different values of $N$, $N_m^i$ and $N_{th}^i$ are plotted in Fig. \ref{fig:test}(a), Fig. \ref{fig:test}(b) and Fig. \ref{fig:test}(c) respectively. In Fig. \ref{fig:test}(a), we plot the energy-efficiency figures for $N=16$,  $32$ by considering $N_m^i=60Kb$ and $N_{th}^i=0.48Mb$. Whereas in Fig. \ref{fig:test}(b) and Fig. \ref{fig:test}(c), the same is plotted for $N_m^i=0.12Mb$, $60Kb$ with $N=16$, $N_{th}^i=0.48Mb$ and  for $N_{th}^i=0.48Mb$, $0.24Mb$ with $N=16$, $N_m^i=0.12Mb$.  It can be observed from the figures that the analytical results match closely with simulation results for all considered scenarios which validate our analytical model. However, at higher load, the simulation results deviate from the analytical results which are due to prediction error. It is interesting to note that this deviation increases with an increase in traffic load. This is due to the following facts. In simulations, Poisson traffic arrivals are predicted by its mean and hence, the prediction error is proportional to the variance of the Poisson distribution. An increment of traffic load enhances the variance of the Poisson distribution and, hence the prediction error.   

\subsubsection{Effect of $N$}
It can be noticed from Fig. \ref{fig:test}(a) that the achieved energy-efficiency figures reduce drastically with an increase in $N$ from $N=16$ to $N=32$. The reason can be explained as follows. In OSMP-EO, if $ONU_i$ enters into the active mode from a sleep mode, it remains active until it up-streams $N_{th}^i$ number of packets. Since in each cycle maximum of $N_m^i$ number of packets get cleared, $\lceil\frac{N_{th}^i}{N_m^i}\rceil$ number of cycles are required to up-stream $N_{th}^i$ number of packets and over this entire period, $ONU_i$ remains active. 
If the value of $N$ is doubled then the duration of each cycle almost become double for fixed allocation scheme and hence, the duration of active periods of $ONU_i$ almost become double resulting in a drastic reduction of energy-efficiency. 
\begin{figure*}[h]
	\centering
	\includegraphics[scale=.6]{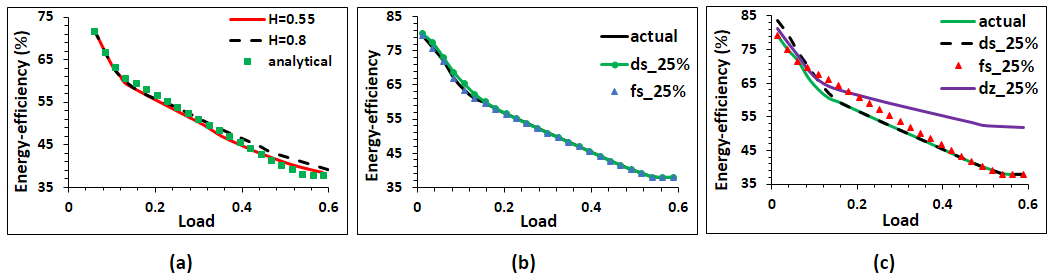}
	\caption{(a) Applicability to self-similar traffic (b) Effect of reducing sleep-to-wake-up time of sleep mode (c) Effect of reducing power consumption of sleep modes.}
	\label{fig:di}
\end{figure*}
\subsubsection{Effect of $N_m^i$}
It can be observed from Fig. \ref{fig:test}(b) that the energy-efficiency figures remain almost same with the variation $N_m^i$ at low load. However, at higher load, a slight decrement in energy-efficiency occurs if $N_m^i$ is incremented. The reason can be explained as follows. If the value of $N_m^i$ is doubled, then the number of cycles over which $ONU_i$ remains in active mode before taking the next decision become half. Whereas, the duration of each cycle become almost double. Thus, the duration of active periods over which the US data is transmitted remains unchanged. However, in OSMP-EO, once an ONU decides to be active from sleep mode, on an average half of a cycle is wasted to receive the next GATE message (refer {Section} \ref{subsubsection:osmp}). 
An increment of $N_m^i$ causes an increase in cycle time and hence, this waste of half of a cycle which results in a reduction of energy-efficiency. Since this waste is present in every occurrence of state transitions from sleep mode to active mode, at low load, when the probability of waking up is very small, the effect of variation of $N_m^i$ is negligible.   
%
\subsubsection{Effect of $N_{th}^i$}
A decrement of $N_{th}^i$ results in a reduction of buffer fill-up time and hence, the probability of state transition from sleep mode to active mode increases. In every occurrence of waking up from $S_m$,  $T_{sw}^{S_m}+T_{cm}/2$ duration is wasted (refer Section \ref{ssec:Smon}). Thus, the decrement of $N_{th}^i$ causes a reduction in energy-efficiency.   
\subsection{Applicability to self similar traffic}
Here, we compare our analytical results with the results, obtained from simulations with self-similar traffic. In order to do so, in Fig. \ref{fig:di}(a), we plot the energy-efficiency figures that are obtained from simulations with self-similar traffic for $H=0.55,~0.8$ and from the analytical model with Poisson traffic as a function of traffic load. Self-similar traffic is generated and predicted in the same way as it is discussed in Section \ref{sec:compari}. We consider $N=16$, $N_m^i=60Kb$ and $N_{th}^i=0.48Mb$. It can be observed that the analytical results match closely with simulation results even for self-similar traffic. However, at high load, the energy-efficiency figures for self-similar traffic is little higher than the analytical results. This is due to the fact that the bursty nature of self-similar traffic increases the packet drop probability and hence, reduces the effective traffic load. The decrement of $H$ reduces the burstiness and wherefore, the packets drop probability. Thus, the results for self-similar traffic approach to the analytical result with the reduction of $H$ which can also be observed from Fig. 
\subsection{Design Insights}
Here, we observe the effect of reducing the sleep-to-wake-up times and the power consumption figures of all sleep modes on the energy-efficiency which can be achieved by designing advance circuitries. For this purpose, energy-efficiency figures that are obtained from the analytical model if the sleep-to-wake-up time and power consumption figures of all sleep modes are reduced individually by $25\%$ are plotted in Fig. \ref{fig:di}(b) and Fig. \ref{fig:di}(c) respectively. It can be noticed from Fig. \ref{fig:di}(b) that the effect of diminishing the sleep-to-wake-up time of all modes are insignificant and hence, makes no sense for a circuit designer to focus on it. Reduction in power consumption of deep sleep and fast sleep increase energy-efficiency at low and moderate load respectively (refer Fig. \ref{fig:di}(c)). This is because, $ONU_i$ enter into the deep sleep if the buffer fill-up time is more than $T_{lb}^{ds}$ which has a significantly high value (refer eq. (\ref{dsv})) and hence, exists only at low load. Similarly, fast sleep exists at moderate load where the effect of reducing $P_{fs}$ is present. However, in OSMP-EO, doze mode reduces the energy consumption of active periods and therefore, the effect of reducing $p_{dz}$ exists in all loads. An increment of traffic load increases the duration of active periods and hence, the enhances the improvement of energy-efficiency. Thus, from the point of view of a design engineer, he should mostly focus on reducing the power consumption figures of the doze mode.
\ref{fig:di}(a). 
\section{Conclusion}
In this paper, we propose an ONU-assisted sleep mode protocol (OSMP-EO) for energy-efficient ONU design in EPON where ONUs are allowed to save energy even during the active cycles. This provides a significant improvement in energy savings especially at high load ($\sim 40\%$) when the energy-efficiency approaches zero in all existing ONU-assisted protocols. A mathematical analysis of the OSMP-EO protocol is also performed. We first prove that this protocol infringes the memoryless property. We then explain that the entire history can be easily captured by an intelligent selection of discrete observation instants and associated state descriptions which allow formulating a DTMC for analyzing energy-efficiency figures. The analytical model is verified with simulations for both Poisson and self-similar traffic. The analysis reveals that the protocol is almost insensitive to sleep-to-wake-up times of all sleep modes. An improvement in energy-efficiency can be achieved at low and moderate load if the power consumption figures can be reduced by designing better circuitries. Further, the energy-efficiency figure of OSMP-EO is highly sensitive to doze mode power consumption figure and hence,  circuit designer should focus on reducing it.         

\ifCLASSOPTIONcaptionsoff
  \newpage
\fi

\bibliographystyle{IEEEtran}

\end{document}